\documentclass[titlepage, 11pt]{article}

\usepackage{setspace}
\usepackage[dvips]{graphicx}
\graphicspath{{grim/} }
\usepackage{mathptmx}
\usepackage{color}
\usepackage{graphicx}
\usepackage{amsmath,amssymb,amsfonts,amsthm,latexsym,eucal}
\usepackage{fullpage}
\usepackage{pdfsync}
\usepackage{mdwlist}
\usepackage{algorithmic}
\usepackage{algorithm}

\newtheorem{definition}{Definition}
\newtheorem{theorem}{Theorem}
\newtheorem{lemma}[theorem]{Lemma}

\newtheorem{corollary}[theorem]{Corollary}

\newcommand{\ignore}[1]{}

\newcommand{\remove}[1]{}

\newcommand{\B}{\vspace*{-\smallskipamount}}

\title{ A Paradigm for Channel Assignment  and Data Migration  in Distributed Systems}
\author{\textsc{Chadi Kari}\\ \\ 
\\Technical Report
\\
\\
\\
Computer Science and Engineering Department \\ University of Connecticut\\
2010
\date{}
\\}

\date{\today}               

\begin{document}
\maketitle

\begin{abstract}
In this manuscript, we consider the problems of channel assignment in wireless networks and data migration in heterogeneous storage systems. We show that a  \emph{soft edge coloring} approach to both problems  gives rigorous approximation guarantees. 

In the channel assignment problem arising in wireless networks, we are given a graph $G = (V, E)$, and the number of wireless cards $C_v$ for each vertex $v$. The constraint $C_v$ limits the number
of channels that edges incident to $v$ can use. We also have the total number of channels $C_G$ available in the network. For a pair of edges incident to a vertex, they are said to be {\em conflicting}
if the channels assigned to them are the same. Our goal is to assign channels (color edges) so that
the number of conflicts is minimized. In this manuscript we  first study the problem for a homogeneous network where $C_v = k$ and $C_G \ge C_v$ for all nodes $v$.  The problem is NP-hard by a reduction from \textsc{Edge coloring} and
we present two combinatorial algorithms for this case.
The first algorithm is based on a distributed greedy method and 
gives a solution with at most $2(1-\frac{1}{k})|E|$ more conflicts than the optimal solution,
which implies a $(2 - \frac{1}{k})$-approximation. 
We also present a soft edge coloring algorithm that  yields at most $2|V|$ more conflicts than the optimal solution.  The approximation ratio is $1 + \frac{|V|}{|E|}$, which gives a ($1 + o(1)$)-factor for dense graphs.  The algorithm generalizes Vizing's algorithm in the sense that it gives the same result as Vizing's algorithm when $k = \Delta + 1$. Moreover, we show that this approximation result is best possible unless $P = NP$. For the case where $C_v = 1$ or $k$,
we show that the problem is NP-hard even when $C_v = 1$ or $2$, and $C_G = 2$, and present two approximation algorithms. The first algorithm is completely combinatorial and has an approximation ratio of $2-\frac{1}{k}$.  We also develop an SDP-based algorithm, producing a solution with an approximation ratio of $1.122$ for $k = 2$, and $2-\Theta(\frac{\ln k}{k})$ in general.
   
In this manuscript, we also consider  the \emph{ data migration} problem in heterogeneous storage systems.  Large-scale storage systems are crucial components in
data-intensive applications such as search engine clusters,
video-on-demand servers, sensor networks, and grid computing.
A storage server typically consists of a set of storage devices.
In such systems, data layouts may need to be reconfigured over time
for load balancing or in the event of system failure/upgrades.
It is critical to migrate data to their target locations
as quickly as possible to obtain the best performance of the system.
Most of the previous results on data migration assume that each storage
node can perform only one data transfer at a time. A storage node,
however, typically can handle multiple transfers simultaneously and
this can reduce the total migration time significantly.
Moreover, storage devices tend to have heterogeneous capabilities
as devices may be added over time due to storage demand increase.
We consider  the \emph{heterogeneous data migration} problem
where we assume that each storage node has different transfer constraint $c_v$,
representing how many \emph{simultaneous} transfers the node can handle.
We develop algorithms to minimize the data migration time. 
We show that it is possible to find
an optimal migration schedule when all $c_v$'s are even. Furthermore, though
the problem is \textsf{NP}-hard in general, we give an efficient    \emph{soft edge coloring} algorithm
that offers a rigorous $(1 + o(1))$-approximation guarantee.

\end{abstract}

\newpage

\doublespacing

\tableofcontents

\singlespacing

\newpage

\section{Soft Edge Coloring }
\label{sec:soft}

\subsection{Introduction}
In a  multi-radio multi-channel wireless network, simultaneous transmissions from nearby nodes over the same wireless channel may \emph{interfere} with each other and as a result can degrade the performance of the network.
One way to overcome this limitation is to assign independent channels (that can be used without interference) to nearby links of the network. However, the number of independent channels that can be employed is usually limited and insufficient and thus conflicts are bound to happen.\\
\B 
Consider the example shown in Figure~\ref{fig:3node}.
If all links use the same channel for transmissions, only one 
pair of nodes may communicate with each other at a time
due to interferences.  However, if there are three channels available 
and each node has two wireless interface cards (so 
it can use two channels), then we may assign a different channel to each link so that all links can be used at the same time. 

\begin{figure}[h]
\begin{center}
    \centerline{\includegraphics[width=1.4in]{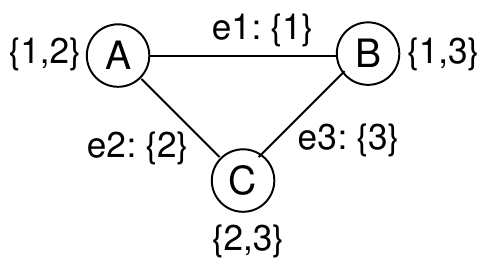}}
    \caption{\small Each node has two wireless interface cards 
(thus can use two different channels) and three channels are available in the network.
	We can assign a distinct channel to each link as shown above so that 
there is no conflict among links.
\label{fig:3node}}
\vspace{-0.3in}
\end{center}
\end{figure}

We informally define the \textsc{Soft edge coloring} for the channel assignment  problem as follows:
We are given a graph $G = (V, E)$, and constraints on the number of 
wireless cards $C_v$ for all $v$. These constraints limit the number
of colors that edges incident to $v$ can use. 
In addition, we have a constraint on the total number
of channels available in the network (denoted as $C_G$).
For a pair of edges incident to a vertex, they are said to be {\em conflicting} 
if the colors assigned to them are the same.
Our goal is to color edges (assign channels) so that
the number of conflicts is minimized while satisfying 
constraints on the number of colors that can be used. 
In section \ref{sec:soft}, we study this problem for homogeneous networks where $C_v = k$
and $C_G \ge C_v$  and for networks  where $C_v = 1$ or $k$ for all nodes $v$. \\

In section \ref{sec:hom} we consider a homogeneous network where $C_v = k$
and $C_G \ge C_v$ for all nodes $v$,  we show 
that the problem is NP-hard for homogeneous networks (section \ref{sec:np}).
We present two combinatorial algorithms in sections \ref{sec:greedy} and \ref{sec:improv} :
First a distributed greedy algorithm that
gives a solution with at most $2(1-\frac{1}{k})|E|$ more conflicts than the optimal solution,
which implies a $(2 - \frac{1}{k})$-approximation. 
The second algorithm  yields at most $2|V|$ more conflicts than the optimal solution. 
The approximation ratio is $1 + \frac{|V|}{|E|}$, which gives a ($1 + o(1)$)-factor for dense graphs. 
The algorithm generalizes Vizing's algorithm in the sense that it gives the same result as Vizing's algorithm when $k = \Delta + 1$. Moreover, we show in section \ref{sec:best} that this approximation result is best possible unless $P = NP$.
\B
For the case where $C_v = 1$ or $k$,
we show in section \ref{sec:hetnp} that the problem is NP-hard even when $C_v = 1$ or $2$, and $C_G = 2$, and present two approximation algorithms.
The first algorithm in section \ref{sec:ext} is completely combinatorial and has an approximation ratio of $2-\frac{1}{k}$. The second is an SDP-based algorithm, producing a solution with an approximation ratio of $1.122$ for $k = 2$, and $2-\Theta(\frac{\ln k}{k})$ in general (section \ref{sec:sdp}).

\subsubsection{Edge coloring}
In the traditional edge coloring problem,
the goal is to find the minimum number 
of colors required to have a proper edge coloring.
The problem 
is $NP$-hard even for cubic graphs \cite{Holyer}.
For a simple graph, a solution using at most $\Delta + 1$ colors
can be found by Vizing's theorem \cite{V64} where
$\Delta$ is the maximum degree of a node.
For multigraphs, there is an approximation algorithm
which uses at most $1.1 \chi' + 0.8$ colors where $\chi'$ is the optimal number of colors 
required \cite{multigraph} (the additive term was improved to $0.7$ by Caprara
 \cite{caprara98improving}). 
Recently, Sanders and Steurer developed an algorithm that
gives a solution with $(1+\epsilon) \chi'+O(1/\epsilon)$ colors \cite{sanders05}.

{\sc Soft edge coloring} is a variant of the {\sc Edge coloring} problem.
In our problem, coloring need not be proper (two adjacent edges
are allowed to use the same color)---the goal is to minimize the
number of such conflicts. 
In addition, each node has its local color constraint, which limits the
number of colors that can be used by the edges incident to the node.
For example, if a node has two wireless cards ($C_v = 2$), 
the node can choose two colors 
and edges incident to the node should use only those two colors.

\subsubsection{Related Work} 
\paragraph{Relationship to \textsc{Min k-partition} and \textsc{Max k-cut}.} 
The {\sc Min k-partition} problem is 
to color vertices with $k$ different colors
so that the total number of conflicts (monochromatic edges) is minimized.
It is the dual of the well-known {\sc Max k-cut} problem \cite{kann:istcs}.
Our problem for homogeneous networks ($C_G = C_v = k$ for all $v$) 
is an edge coloring version of {\sc Min $k$-parition} problem\footnote{
Or it can be considered as {\sc Min $k$-partition} problem  
when the given graph is a line graph where 
the line graph of $G$ has a vertex corresponding to
each edge of $G$, and there is an edge between two vertices
in the line graph if the corresponding edges
are incident on a common vertex in $G$.}.
Kann \cite{kann:istcs} showed that for $k > 2$ and for every $\epsilon > 0$,
there exists a constant $\alpha$ such that the {\sc Min $k$-partition} 
cannot be approximated within a constant factor
unless $P = NP$
\footnote{Their objective function is
slightly different from ours as they do not count self-conflicts. Their inapproximability bound 
of O($|V|^{2 - \epsilon}$) can be extended to the bound of O($|V|^{1-\epsilon}$) in our objective function.}. 
%

\paragraph{Other Related Work.}
Fitzpatrick and Meertens~\cite{softexp} have considered a variant of 
graph coloring problem (called the {\sc Soft graph coloring} problem)
where the objective is to develop a distributed  algorithm for coloring
vertices so that the number of conflicts is minimized.
The algorithm repeatedly recolors vertices to quickly reduce 
the conflicts to an acceptable level.
They have studied experimental performance for regular graphs 
but no theoretical analysis has been provided.   
Damaschke~\cite{softpath} presented a distributed soft coloring algorithm
for special cases such as paths and grids, 
and provided the analysis on the number of conflicts as a function of time $t$. 
In particular, the conflict density on the path is given as $O(1/t)$ 
when two colors are used,
where the conflict density is the number of conflicts divided by $|E|$.

\subsubsection{Problem Definition}
We are given a graph $G=(V, E)$ representing a wireless network, where $v \in V$ represents a node in the wireless network and an edge $e = (u, v) \in E$ represents a communication link
between $u$ and $v$.
Each node $v$ can use $C_v$ different channels and the total number
of channels that can be used in the network is $C_G$. 
More formally, let $E(v)$ be the set of edges incident to $v$ and $c(e)$ be the channel assigned to $e$. 
Then $|\bigcup_{e \in E(v)} \{c(e)\}| \le C_v$ and 
$|\bigcup_{e \in E} \{c(e)\}| \le C_G$.

A pair of edges $e_1$ and $e_2$ in $E(v)$ are said to be conflicting
if the two edges use the same channel.
Let us define the {\em conflict number}, $CF_e (v)$ of an edge 
$e \in E$ at a vertex $v$ to be the number of edges (including $e$) that conflict with $e$ at $v$.
In other words, for an edge $e$ incident to $v$,
$CF_e (v)$ is the number of edges in $E(v)$  that use the same channel as $e$. 
Our goal is to minimize the total number of conflicts. That is, 
\begin{equation} \label{eq:metric} 
CF_G =  \sum_{e = (u, v)\in E}(CF_e (u) + CF_e (v)) . 
\end{equation}
Note that in (\ref{eq:metric}) each conflict is counted twice. We can also define the total number of conflicts as the sum of the squares of the color classes at each node.
That is, let $E_i(v)$ be the set of  edges with color $i$ at node $v$. Then,  
\begin{equation} \label{eq:otherdef} 
CF_G = \sum_{v\in V}\sum_i{|E_i(v)|^2} .
\end{equation}
The two objective functions are equivalent.
Note that the number of conflicts at a vertex $v$, $\sum_i{|E_i(v)|^2}$, 
is minimized locally when edges in $E(v)$ are distributed evenly to each color.
Figure \ref{fig:cfexample} shows a feasible coloring and  the number of conflicts for the given graph.
\begin{figure}[ht]
\begin{center}
    \centerline{\includegraphics[width=4in]{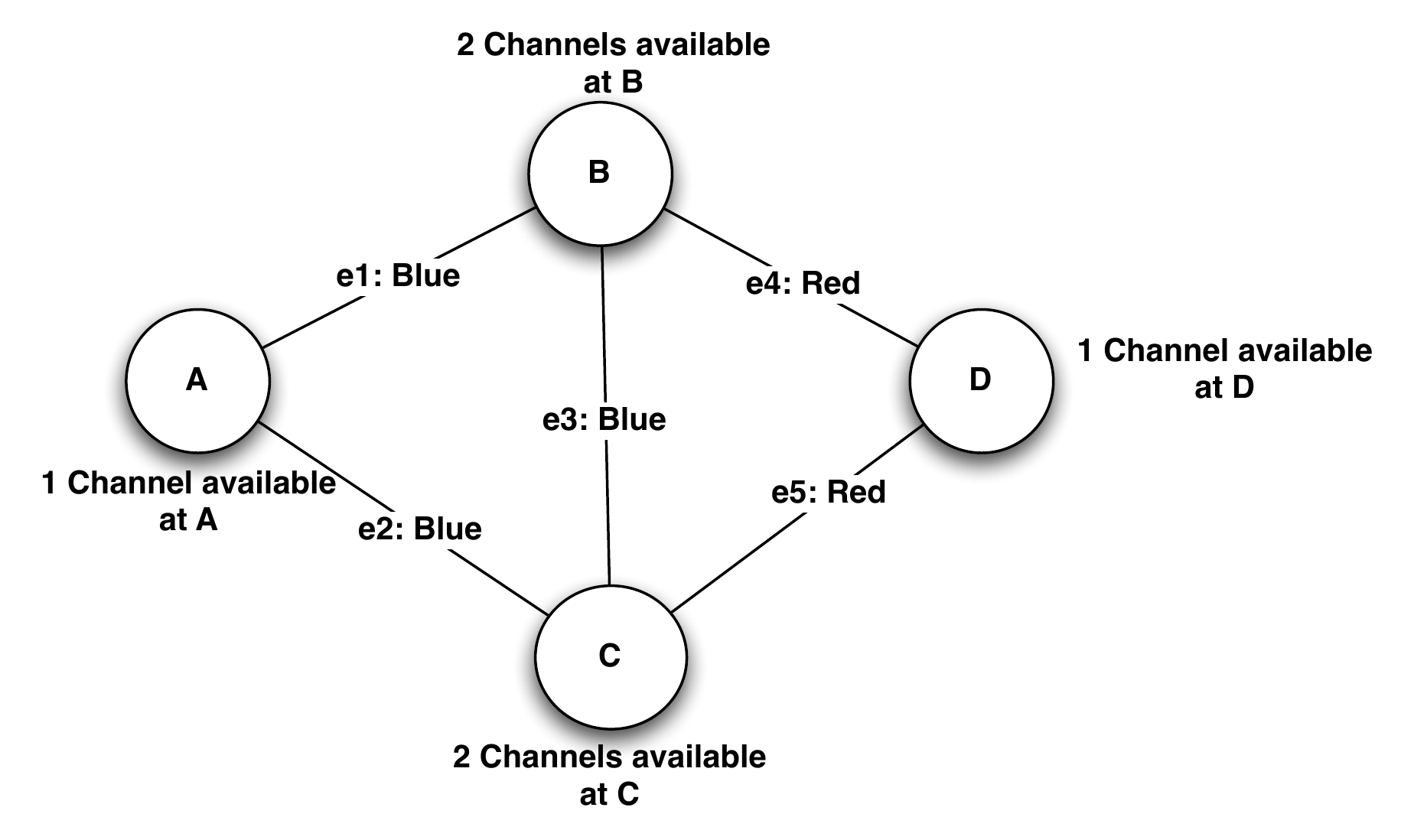}}
    \caption{
$CF(e_1) = 4$ ($2$ conflicts at $A$ and $2$ conflicts at $B$), $CF(e_4) = 3$ ($1$ conflict at $B$ and $2$ conflicts at $D$), $CF(e_2) = CF(e_3) = 4$ and $CF(e_5)=3$.
\newline
Total number of conflicts is $CF_G=18$.
\label{fig:cfexample}}
\end{center}   
\end{figure} 

In sections \ref{sec:hom} and \ref{sec:improv}, we denote {\em channels} by {\em
colors} and use edge coloring and channel assignment, interchangeably.
We also use conflicts and interferences interchangeably.

\subsection{Algorithms for Homogeneous Networks}
\label{sec:hom}

In this section, we consider the case for a homogeneous network where for all nodes v, the number of channels that can be used is the same ($C_v = k$). 
\subsubsection{NP hardness}
\label{sec:np}
For an arbitrary $k$, the problem is NP-hard as the edge coloring problem can be reduced to our problem by setting $k = C_G = \Delta$ where $\Delta$ is the maximum degree of nodes.

\subsubsection{Greedy Algorithm} 
\label{sec:greedy}
The greedy algorithm works as follows:
We choose colors from $\{1, \dots, k\}$ 
(
We only use $k$ colors even for problem instances where $C_G > k$ as $C_G = k$ is the worst case.) 
For any  uncolored edge $e = (u, v)$,
we choose a color for edge $e$ that introduces 
the smallest number of conflicts.
More formally, when we assign a color to $e = (u, v)$,
we count the number of edges in $E(u) \bigcup E(v)$
that are already colored with $c$ (denoted as $n(c, e)$),
and choose color $c$ with the smallest $n(c, e)$, ties are broken arbitrarily. 

\begin{algorithm}[h]
\caption{\bf  Greedy Algorithm}
\label{dgreedy}
\begin{algorithmic}
\FOR {each edge $e = (u, v)$}
	\FOR {each color $i$}
	\STATE compute the number of edges in $E(u)$ and $E(v)$ using color $i$.
	\ENDFOR
	\STATE let $c$ be the color with min $n(i, e)$ for all colors $i$.
	\STATE assign color $c$ to edge $e$.

\ENDFOR
\end{algorithmic}
\end{algorithm}

\begin{theorem}
\label{th:onehop-greedy}
The greedy algorithm yields at most  $2(1 -\frac{1}{k})|E|$ conflicts more than the optimal solution
in homogeneous networks, which implies  a  $(2 - \frac{1}{k})$-approximation.
\end{theorem}
To prove Theorem \ref{th:onehop-greedy} we need to show the following two lemmas. We first obtain a lowerbound on the optimal solution.

\begin{lemma}
The total number of conflicts when $C_v = k$ for all nodes $v$ in any channel assignment
is at least
$    \sum_{v}\frac{d_v^2} {k} $.
\label{lemma:lower}
\end{lemma}

The second lemma gives an upperbound on the number of conflicts in our solution.
\begin{lemma}
The total number of conflicts introduced by the greedy algorithm is at most 
$    \sum_{v}\frac{d_v^2} {k} + 2(1 -\frac{1}{k})|E|$.
\label{lemma:upper}
\end{lemma}
Note that the algorithm can be performed in a distributed manner as each node needs only local information. 

\noindent{\em Remark 1:} we can consider a simple randomized algorithm, in which each edge
chooses its color uniformly at random from $\{1, \dots, k\}$. 
The algorithm gives the same expected approximation guarantee and
it can be easily derandomized using conditional expectations.  \\

\subsubsection{Improved Algorithm}
\label{sec:improv}
In this section, we give an algorithm with an additive factor of 2$|V|$ and an approximation ratio of  $1 + \frac{|V|}{|E|}$. Our algorithm is a generalization of Vizing's algorithm in the sense that
it gives the same result as Vizing's algorithm  when $k = \Delta + 1$ where $\Delta$ is the maximum degree of nodes. 
We first define some notations. For each vertex $v$, let $m_v = \lfloor \frac{d_v}{k} \rfloor$ and $\alpha_v = d_v  -  m_vk $.

Let $|E_i(v)|$ be the size of the color class of color $i$ at vertex $v$ i.e. the number of edges adjacent to $v$ that have color $i$.

\begin{definition}
A color $i$ is called \emph{strong} on a vertex $v$ if $|E_i(v)|= m_v + 1$. A color $i$ is called weak on $v$ 
if $|E_i(v)| =  m_v$ . A color $i$ is called very weak on $v$ if $|E_i(v)| < m_v$.
\end{definition}
\begin{definition}
A vertex $v$ has a \emph{balanced} coloring if the number of strong classes at $v$ is at most $\min(\alpha_v + 1, k-1)$ and no color class in $E(v)$ is larger than $m_v + 1$. A graph $G = (V, E)$ has a balanced coloring if each vertex $v \in V$ has a balanced coloring.   
\end{definition}

The intuition behind the definition of balanced coloring is that the local number of conflicts at a vertex is minimized when
edges are distributed as evenly as possible to each color. We try to achieve the balanced coloring
by not creating too many strong color classes and also allowing at most one more strong color class than the optimal solution.
In the following we present an algorithm that achieves a balanced coloring for a given graph $G = (V, E)$;  we show in Theorem~\ref{approxfactor} that a balanced coloring implies an additive approximation factor of 2$|V|$ in terms of number of conflicts and an approximation ratio of $1 + \frac{|V|}{|E|}$. 

In Algorithm {\sc BalancedColoring}($e$) described below, we color edge $e$ so that the graph has a balanced coloring 
(which may require the recoloring of already colored edges to maintain the balanced coloring),
assuming that it had a balanced coloring before coloring $e$.
We perform {\sc BalanacedColoring} for all edges in arbitrary order.
The following terms are used in the algorithm description.
Let $|S_v|$ denote the number of strong color classes at vertex $v$.
\begin{definition}\label{def:weak}
For vertex $v \in V$ with $|S_v| < \min(\alpha_v + 1, k - 1)$ or with $|S_v| = k - 1$, $i$ is a missing color if $i$ is weak or very weak on $v$.  For vertex $v \in V$ with $|S_v| = \alpha_v + 1$, $i$ is a missing color if $i$ is very weak on $v$
\end{definition}

In Lemma \ref{missing}, we will show that it is safe to use a missing color at a vertex 
for an edge incident to it (i.e., we can maintain the balanced coloring property).

\begin{definition}\label{def:path}
An \emph{$ab$-path} between vertices $u$ and $v$ where $a$ and $b$ are colors, is a path connecting $u$ and $v$ and has the following properties:
\begin{itemize}
\item Edges in the path have alternating colors $a$ and $b$. 
\item Let $e_1 = (u,w_1)$ be the first edge on that path and suppose $e_1$ is colored $a$, then $u$ must be missing $b$ and not missing $a$. 
\item If $v$ is reached by an edge colored $b$ then $v$ must be missing $a$ but not missing $b$, otherwise if $v$ is reached by an edge colored $a$ then $v$ must be missing $b$ and not missing $a$.
\end{itemize}  
\end{definition}
\begin{definition}
A flipping of an $ab$-path is a recoloring of the edges on the path such that edges previously with color $a$ will be recolored with color $b$ and vice versa. 
\end{definition}

Note that an $ab$-path is not necessarily a simple path and may contain a cycle as a vertex can have multiple edges with the same color.
We show that flipping an $ab$-path does not violate the balanced coloring property in Lemma \ref{flipping}.
Algorithm {\sc BalancedColroing} works as follows. \\

\begin{figure}[t]
\begin{center}
    \centerline{\includegraphics{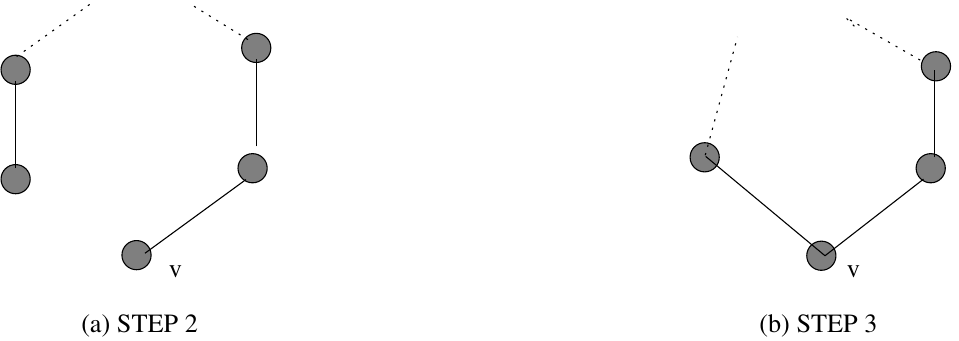}}
    \caption{\small The figures illustrate how recoloring is performed in {\sc BalancedColoring}. The colors beside edges
indicate the original color and the color after recoloring. \vspace{-0.3in}}
 \label{fig:phaseone}
 \end{center}
\end{figure}

\noindent Algorithm {\sc BalancedColoring($e = (v, w)$}) \\
Let $w_1 = w$. At $i$-th round ($i = 1, 2, \dots$), we do the following. \\
\underline{STEP 1:} 
Let $\mathcal{C}_v$ be the set of missing colors on $v$.
If $i=1$, $\mathcal{C}_{w_1}$ is the set of missing colors on $w_1$. 
When $i \ge 2$, ${\mathcal{C}_{w_i}}$ is  the set of missing colors on $w_i$ excluding color $c_{w_{i-1}}$.
($c_{w_{i-1}}$ is defined in STEP 2 at $(i-1)$-th round).
If $\mathcal{C}_v \cap \mathcal{C}_{w_i} \ne \emptyset$, 
then choose a color $c \in \mathcal{C}_v \cap \mathcal{C}_{w_i}$, color edge $(v,w_i)$ with $c$ and terminate. \\
\underline{STEP 2:} If $\mathcal{C}_v \cap \mathcal{C}_{w_i} = \emptyset$, choose $c_v \in \mathcal{C}_v$ and $c_{w_i} \in \mathcal{C}_{w_i}$ ($C_v \neq \emptyset$ and  $\mathcal{C}_{w_1} \neq \emptyset$ by lemma \ref{saturated} ) . Find a $c_vc_{w_i}$-path that starts at $w_i$ and ends at a vertex other than $v$. If such a path exists, flip this path, color edge $(v,w_i)$ with $c_v$ and terminate (Fig. \ref{fig:phaseone} a).\\
\underline{STEP 3:} If all $c_vc_{w_i}$-paths that start at vertex $w_i$ end at $v$, fix one path and let $(v,w_{i+1})$ be the last edge on that path. The edge $(v,w_{i+1})$ must have color $c_{w_i}$ by definition \ref{def:path} . Uncolor it and color edge $(v,w_i)$ with $c_{w_i}$ (Fig. \ref{fig:phaseone} b). Mark edge $(v,w_i)$ as ``used" and go to $(i+1)$-th round and repeat the above steps with edge $(v, w_{i+1})$.

\paragraph{Analysis} In the following, we prove that our algorithm terminates and achieves 
a balanced coloring. First Lemma \ref{saturated} and \ref{lemma:saturated2} show that we can always find a missing color at each round and at Lemma \ref{termin} shows that at some round $j < d_v$,
the algorithm terminates. Due to the choice of missing colors and $ab$-path,
we can show that our algorithm gives a balanced coloring (Lemma \ref{missing} and \ref{flipping} ).

\begin{lemma}\label{saturated}
For a given edge $(v, w_1)$, there is a missing color at $v$ and $w_1$.
That is, $C_v \neq \emptyset$ and  $\mathcal{C}_{w_1} \neq \emptyset$.
\end{lemma}

For $w_i$, $i \ge 2$,  we need to choose a missing color at $w_i$ other than $c_{w_{i-1}}$. We prove in the following lemma
that there is a missing color other than $c_{w_{i-1}}$. 
\begin{lemma}\label{lemma:saturated2}
At $i$-th round ($i \ge 2$), there is a missing color other than $c_{w_{i-1}}$ at $w_i$. 
\end{lemma}

\begin{lemma}\label{termin}
At some round $j < d_v$, there exists a $c_vc_{w_j}$-path starting at $w_j$ and not ending at $v$.
\end{lemma}

\begin{figure}[ht]      
\begin{center}        
  \begin{center}
    \centerline{\includegraphics[width=1.5in]{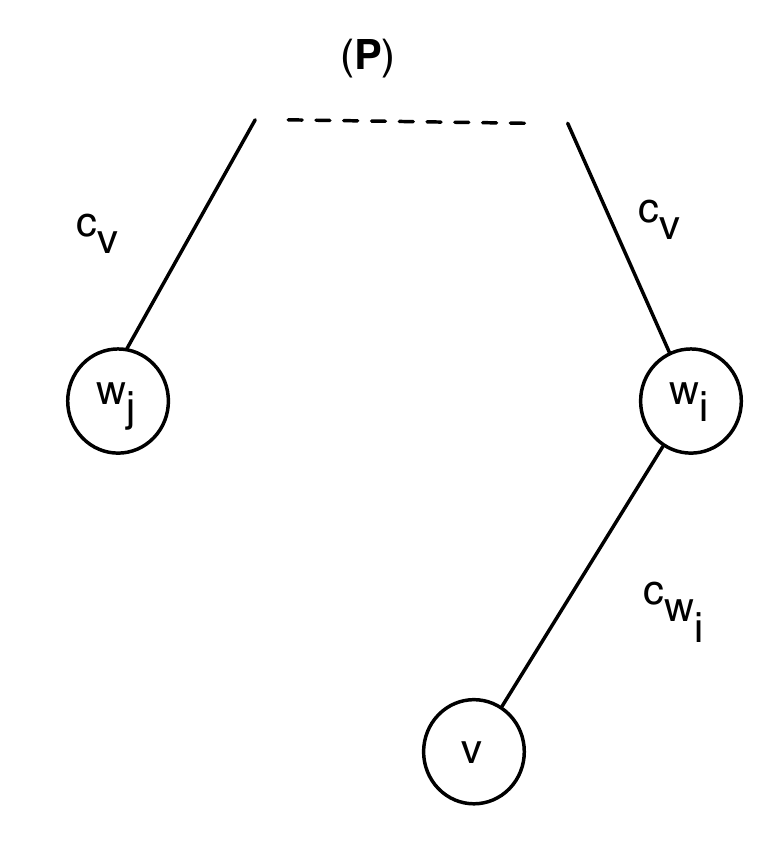}}
    \caption{$c_{w_i} = c_{w_j}, $if the $c_vc_{w_i}$-path $P$ connecting $v$ and $w_j$ exists then $P - (v,w_i)$ is a $c_vc_{w_i}$-path connecting $w_i$ and $w_j$, so the algorithm would terminate at STEP 2 in round $i < j$.
\label{fig:lemma12}}
  \end{center}
\end{center}
\end{figure}

\begin{lemma}\label{missing}
Let $v$ be a vertex that has a balanced coloring. Let $e \in E(v)$ be uncolored and let $i$ be a missing color on $v$. Coloring $e$ with $i$ will not violate the balanced coloring property at $v$. 
\end{lemma}

\begin{lemma}\label{flipping}
A flipping of an $ab$-path in a graph with balanced coloring will not violate the balanced coloring.
Moreover, a terminal node of the path which was originally missing $b$ (resp., $a$) and not missing $a$ (resp., $b$)
will be missing $a$ (resp., $b$) 
after flipping.
\end{lemma}

\begin{theorem}\label{mainthm}
The above algorithm terminates and achieves a balanced coloring.
\end{theorem}

\begin{theorem}\label{approxfactor}
A balanced coloring of a graph gives at most $2|V|$ more conflicts than $OPT$ which implies a $(1 + \frac{|V|}{|E|})$-approximation algorithm for the soft edge coloring problem in homogeneous networks.
\end{theorem}

\begin{corollary}
For any $v$ if $\alpha_v = k-1$  the algorithm gives an optimal solution.
\end{corollary}

\subsubsection{Best Possible approximation for dense graphs unless P = NP}
\label{sec:best}
 We can show that the approximation ratio given by the algorithm 
is best possible unless $P=NP$. 
\begin{theorem}
For a given constant $0 < \epsilon< 1$, it is NP-hard to approximate the channel assignment problem  in homogeneous networks 
within an additive term of $o(2|V|^{1- \epsilon})$ and thus it is NP-hard to get an approximation factor with $1 + o(\frac{|V|^{1- \epsilon}}{|E|})$ .
\label{theorem:inapprox}
\end{theorem}

\subsection{Networks where $C_v = 1$ or $k$}
\label{sec:het}
In this section, 
we present two algorithms for networks with $C_v = 1$ or $k$
 and analyze the approximation ratios of the algorithms.
The case where $C_v = 1$ or $k$ is interesting since 
 it reflects a realistic setting, in which 
most of mobile stations are equipped with one wireless card and 
nodes with multiple wireless cards are placed in strategic places
to increase the capacity of networks.

\subsubsection{NP-Hardness}
\label{sec:hetnp}
The problem is NP-hard even when $C_v = 1$ or $2$. We show it by reducing 3SAT to this problem. 
\begin{theorem}
The channel assignment problem to minimize the number of conflicts 
is NP-hard even when $C_v =1$ or $2$, and $C_G = 2$.
\label{theorem:nphard}
\end{theorem}

\subsubsection{Extended Greedy Algorithm}
\label{sec:ext}

Here we present an extended greedy algorithm when $C_v = 1$ or $k$,
and $C_G \ge k$.
The approximation factor is $2 - \frac{1}{k}$. Even though
the algorithm based on SDP (semi-definite programming) gives
a better approximation factor (see Section \ref{sec:sdp}), 
the greedy approach gives a simple combinatorial algorithm. 
The algorithm generalizes the idea of the greedy algorithm 
for homogeneous networks.

Before describing the algorithm, we define some notations. 
Let $V_i \subseteq V$ be the set of nodes $v$ with $C_v = i$
(i.e., we have $V_1$ and $V_k$).
$V_1$ consists of connected clusters $V^1_1, V^2_1, \dots V^t_1$,
such that nodes $u, v \in V_1$ belong to the same cluster
if and only if there is a path composed of nodes in $V_1$ only.
(See Figure \ref{fig:cgreedy} for example.)
Let $E^i_1$  be a set of edges both of which endpoints are in $V^i_1$.
We also define $B^i_1$ to be a set of edges whose one
endpoint is in $V^i_1$ and the other is in $V_k$. We can think of 
$B^i_1$ as a set of edges in the boundary of cluster $V_1^i$.
Note that all edges in $E^i_1 \bigcup B^i_1$ should have the same color.
$E_k$  is a set of edges both of which endpoints are in $V_k$.
$E_1$ is defined to be $\bigcup_i E^i_1$ and $B_1$ is defined to be $\bigcup_i B^i_1$

\begin{figure}[ht]      
\begin{center}        
  \begin{center}
    \centerline{\includegraphics[width=2.5in]{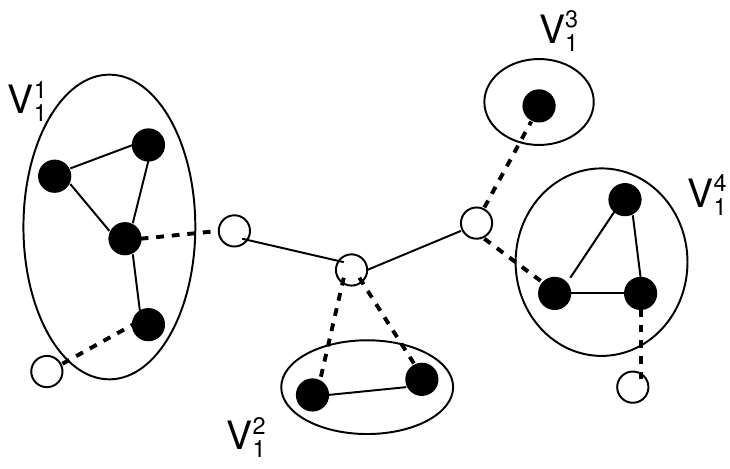}}
    \caption{The figure shows an example of clusters $V_1^i$
when $C_v = 1$ or $k$. Black nodes have only one wireless card
and white nodes have $k$ wireless cards. Dotted lines belong
to $B^i_1$.
\label{fig:cgreedy}}
  \end{center}
\end{center}
\end{figure}

In the greedy algorithm for homogeneous networks,
each edge greedily chooses a color so that
the number of conflicts it creates (locally) is minimized.
Similarly, when $C_v = 1$ or $k$, edges in the same cluster $V_1^i$
choose a color  so that the number of conflicts it creates is minimized.
Formally, we choose a color $c$ with minimum value of 
$\sum_{e = (u, v) \in B^i_1, v \in V_k} n_c(v)$ where
$n_c(v)$ is the number of edges $e' \in E(v)$ with color $c$. Algorithm \ref{cgreedy} describes the extended greedy algorithm.

\begin{algorithm}[ht]
\caption{\bf  Extended Greedy Algorithm}
\label{cgreedy}
\begin{algorithmic}
\FOR {each cluster $V_1^i$}
	\STATE (choose a color for edges in $E^i_1 \bigcup B^i_1$ as follows)

\IF {$B^i_1$ is empty} 
	\STATE choose any color for $E^i_1$.
\ELSE 
	\FOR{each color $c \in \{1, \dots, k\}$ }
	\STATE count the number of conflicts to be created 
when we choose color $c$ for $E^i_1 \bigcup B^i_1$.
Formally, count $\sum_{e = (u, v) \in B^i_1, v \in V_k} n_c(v)$ where
$n_c(v)$ is the number of edges $e' \in E(v)$ 
with color $c$.
\ENDFOR
\STATE choose a color $c$ 
that  minimizes $\sum_{e = (u, v) \in B^i_1, v \in V_k} n_c(v)$. 
\ENDIF
\ENDFOR

\FOR{each edge that belongs to $E_k$}
\STATE choose a color using the greedy algorithm
in Section \ref{section:greedy}.
\ENDFOR
\end{algorithmic}
\end{algorithm}

Any edges $(u, v)$  incident to a vertex in $V_1$
should use the same color 
and therefore are conflicting with each other no matter what algorithm we use.
Given an optimal solution, consider $OPT(V_1)$ and $OPT(V_k)$
where $OPT(S)$ is the number of conflicts at vertices in $S \subseteq V$.
Similarly, we have $CF(V_1)$ and $CF(V_k)$ where
$CF(S)$ is the number of conflicts at vertices in $S \subseteq V$ in our solution.
Then we have $OPT(V_1) = CF(V_1)$. Therefore, we only need to
compare $OPT(V_k)$ and $CF(V_k)$

\begin{theorem}
The approximation ratio of the extended greedy algorithm 
at $V_k$ is  $2 - \frac{1}{k}$.
\label{th:1k-approx}
\end{theorem}

Note that as in the homogeneous case, 
we can obtain the same expected approximation guarantee with a randomized algorithm, 
i.e., choose a color uniformly at random for each cluster $V_1^i$.
Note also that the approximation ratio remains the same for any $C_G \ge k$.
In the following section, we obtain a slightly better approximation factor 
using SDP relaxation  when $C_v = 1$ or $k$ and $C_G = k$.

\subsubsection{SDP-based Algorithm}
\label{sec:sdp}

In this subsection, we assume that $k$ different channels are
available in the network and all nodes have $1$ or $k$ wireless cards.  We formulate the problem using semidefinite programming.
Consider the following vector program (VP), which we can  
convert to an SDP and obtain an optimal solution in polynomial time. 
We have an $m$-dimensional unit vector $Y_e$ for each edge $e$ ($m \le n$).
\begin{eqnarray}
\mbox{\bf VP:~~~~~~}
\min  \sum_v \sum_{e_i, e_j \in E(v)} \frac{1}{k}((k-1) Y_{e_i} \cdot Y_{e_j}& +& 1)
\label{eqn:obj}\\
|Y_e| &= & 1 \label{eqn:edge} \\
Y_{e_i} \cdot Y_{e_j} &=& 1  ~~ \mbox{if~} C_v = 1 \label{eqn:color1}, ~e_i, e_j \in E(v) \\
Y_{e_i} \cdot Y_{e_j} &\ge& \frac{-1}{k-1}  ~~ \mbox{for ~}  ~e_i, e_j \in E(v) 
\end{eqnarray}

We can relate a solution of VP to a channel assignment as follows.
Consider $k$ unit length vectors in $m$-dimensional space such 
that for any pair of vectors $v_i$ and $v_j$, 
the dot product of the vectors is $- \frac{1}{k-1}$.
(these $k$ vectors form an equilateral $k$-simplex on a $(k-1)$-dimensional space
\cite{kcut,kms94}.)
Given an optimal channel assignment of the problem,
we can map each channel to a vector $v_i$.
$Y_e$ takes the vector that corresponds to the channel of edge $e$.
If $C_v = 1$, all edges incident to $v$ should have the same color.
The objective function is exactly the same as the number of 
conflicts in the given channel assignment
since if $Y_{e_1} = Y_{e_2}$ ($e_1$ and $e_2$ have the same color),
it contributes $1$ to the objective function, and 0 otherwise.
Thus the optimal solution of the VP gives a lower bound on the optimal solution.

The above VP can be converted to a semidefinite program (SDP) and
solved in polynomial time (within any desired precision) 
\cite{ali95,gls81,gls87,nn90,nn94}, and given a solution for the SDP,
we can find a solution to the corresponding VP, 
using incomplete Cholesky decomposition \cite{golub83}. 

We use the rounding technique used for {\sc Maxcut} by Goeman and Williamson 
\cite{GW} when $k = 2$ and show that the expected number of conflicts 
in the solution 
is at most $1.122OPT$.
 When $k > 2$, we obtain  the approximation guarantee of  $ 2 - \frac{1}{k} - \frac{2(1 + \epsilon) \ln k}{k} +
 O(\frac{k}{(k-1)^2})$ 
  where $\epsilon(k)\sim \frac{\ln \ln k}{(\ln k)^\frac{1}{2}} $.\\

\section{Data Migration}

\subsection{Introduction}

Large-scale storage systems are crucial components for today's data-intensive 
applications such as search engine clusters, video-on-demand servers, 
sensor networks, and grid computing.
A storage cluster can consist of several hundreds to thousands of storage devices,
which are typically 
connected using a dedicated high-speed network. 
In such systems, data locations may have to be changed over time 
for load balancing or in the event of disk addition and removal which can occur freuqently \cite{UCSB}.
It is critical to migrate data to their target disks
as quickly as possible to obtain the best performance of the system
since the storage system will perform sub-optimally  
until migrations are finished.

The data migration problem can be informally defined as follows.
We have a set of disks $v_1, v_2, \dots, v_n$ and
a set of data items $i_1, i_2, \dots, i_m$. Initially,
each disk stores a subset of items. Over time, data items
may be moved to another disk for load balancing or for system reconfiguration.
We can construct a \emph{transfer graph} $G = (V, E)$ where 
each node represents a disk and an edge $e= (u, v)$ represents a data item
to be moved from disk $u$ to $v$. 
Note that the transfer graph can be a multi-graph (i.e., there
can be multiple edges between two nodes) when multiple data
items are to be moved from one disk to another.
See Figure~\ref{fig:dm-example} for example.
In their ground-breaking work, Hall ~\cite{Karlin} studied
the data migration problem of scheduling migrations 
and developed efficient approximation algorithms.
In their algorithm, they assume that each disk can participate
in only one migration at a time and  both disks and data items are identical;
they show that this is exactly the problem of edge-coloring the transfer graph. 
Algorithms for edge-coloring multigraphs can now be applied to produce 
a migration schedule since each color class represents a matching 
in the graph that 
can be scheduled simultaneously.
\begin{figure}[h]
\centering
\includegraphics[width=2.5in,height=1.5in]{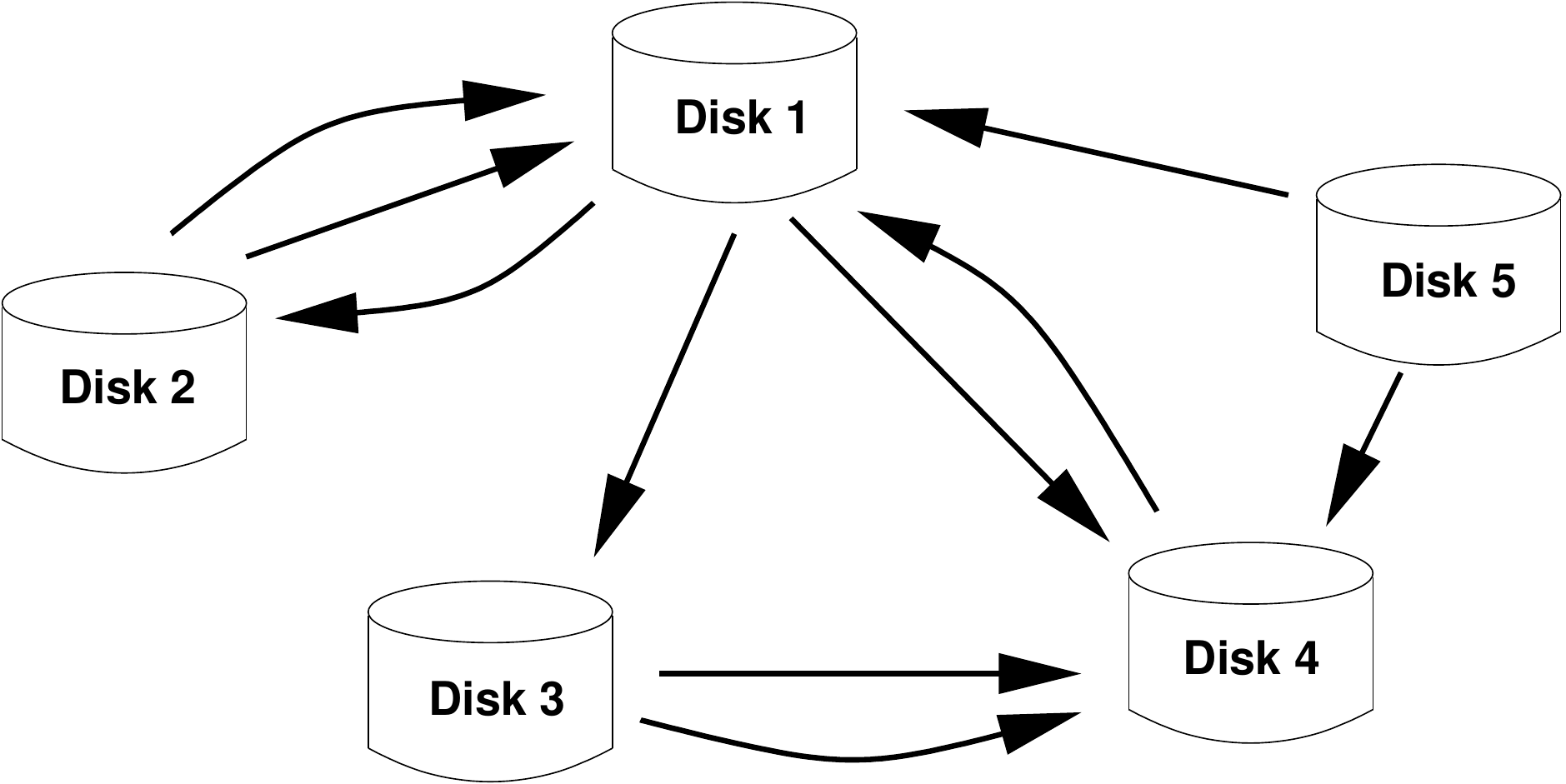}
\vspace{0.2in}
\caption{
An example of data transfer instance}
\label{fig:dm-example}
\end{figure}

\subsection{Related Work}
Hall et al ~\cite{Karlin} studied
the problem of scheduling migrations given a set of disks, with each storing a subset of items
and a specified set of migrations. A crucial constraint in their problem is that each disk can participate
in only one migration at a time. If both disks and data items are identical,
this is exactly the problem of edge-coloring a multi-graph.
That is, we can create a transfer graph $G(V, E)$ that has a node corresponding to each disk, and
a directed edge corresponding to each migration that is specified.
Algorithms for edge-coloring multigraphs can now be applied to produce
a migration schedule since each color class represents a matching
in the graph that can be scheduled simultaneously.
Computing a solution with the minimum number of
colors is \textsf{NP}-hard \cite{holyerNPhard}, but several approximation algorithms
are available for edge coloring

\subsection{Problem Definition}
In the {\sc Heterogeneous Data Migration} problem, 
we are given a transfer graph $G = (V, E)$.
Each node in $V$ represents a disk in the system
and each edge $e = (i, j)$ in $E$ represents
a data item that need to be transferred from $i$ to $j$.
We assume that each data item has the same length, and therefore
it takes the same amount of time for each data to migrate.
Note that the resulting graph is a multi-graph as
there may be several data items to be sent from disk $i$ to disk $j$.

We assume that transfers between disks can be done through 
a very fast network connection dedicated to support a storage system.
Therefore, we assume that any two disks can send data to each other directly.
{\em In particular, we assume that each disk $v$ can handle multiple transfers
at a time. Transfer constraint $c_v$ represents 
how many parallel data transfers the disk $v$ can perform
simultaneously}.

Our objective is to minimize the number of rounds
to finish all data migrations.

\subsubsection{Lower Bounds}
We have the following two lower bounds on the optimal solution.
\begin{eqnarray}
LB_1 & = & \Delta' = \max_v \lceil d_v / c_v \rceil \\
LB_2 & = & \Gamma' = \max_{S \subseteq V} \frac{|E(S)|}
{\lfloor \frac{\sum_{v \in S} c_v}{2} \rfloor}  
\end{eqnarray}
where $E(S)$ is the set of edges both of which endpoints
are in $S$.

$LB_1$ follows from the fact that for a node $v$, at most $c_v$ data items can be migrated
in a round.  When all $c_v$'s are even, $LB_1 \le LB_2$ and, 
in fact, we show that there is a migration schedule that can be
completed in $LB_1$ rounds.
The following lemma proves that $LB_2$ is a lower bound on the optimal
solution.
\begin{lemma}
$LB_2$ is a lower bound on the optimal solution.
\end{lemma}
\begin{proof}
An optimal migration is a decomposition of edges in $E$ into 
$E_1, E_2, \dots, E_k$ such that for each $E_i$ and a vertex $v$,
there is at most $c_v$ edges incident to $v$ in $E_i$.
For a subset $S \subseteq V$, let $E_i(S)$ be
the set of edges in $E$ both of which endpoints are in $S$
and $d_i(v, S)$ be the number of edges in $E_i(S)$ incident to $v$.
Then $2 |E_i(S)| =  \sum_{v \in S} d_i(v, S)$.
As $d_i(v,S) \le c_v$, we have $ |E_i(S)| \le 
\lfloor \frac{\sum_{v \in S} c_v}{2} \rfloor $.
As $E_i$'s cover all edges in $E(S)$, the lemma follows.
\end{proof}

\subsection{Optimal Migration Schedule for Even Transfer Constraints}

In this section, we describe a polynomial time algorithm to find an optimal migration
schedule when each node $v$ has even transfer constraint $c_v$. 
We show that it is possible to obtain a migration schedule using $\Delta'$ rounds.

\subsubsection{Outline of Algorithm}
We first present the outline of our algorithm when $c_v$ is even for any $v$.

\begin{enumerate}
\item[(1)] Construct $G'$ so that every node has degree exactly $c_v \Delta'$   
by adding dummy edges.
\item[(2)] Find a Euler cycle (EC) on $G'$. 
\item[(3)] Construct a bipartite graph $H$ by considering the directions of edges
obtained in $EC$. That is, for each node $v$ in $G'$, create two copies $v_{in}$
and $v_{out}$. For an edge $e = (u, v)$ in $G'$, if the edge is visited
from $u$ to $v$ in $EC$, then create an edge from $u_{out}$ to $v_{in}$ in $H$.
\item[(4)] We now decompose $H$ into $\Delta'$ components
by repeatedly finding a $c_v/2$-matching in $H$. 
\item[(5)] Let $M_1, M_2, \dots, M_{\Delta'}$ be the matching
obtained in Step (4). Then schedule one matching in each round.
\end{enumerate}

\subsubsection{Description and Analysis}\label{sec:detail_even}
We now describe the details and show that
the algorithm gives an optimal migration schedule.
\vspace{0.1in}
\noindent{Step (4):}
We now find a $c_v/2$-matching in $H$  where exactly $c_v/2$ edges are matched for each $v_{in}$
and $v_{out}$. We show the following lemma.
\begin{figure}[h]
\centering
\includegraphics[width=3in]{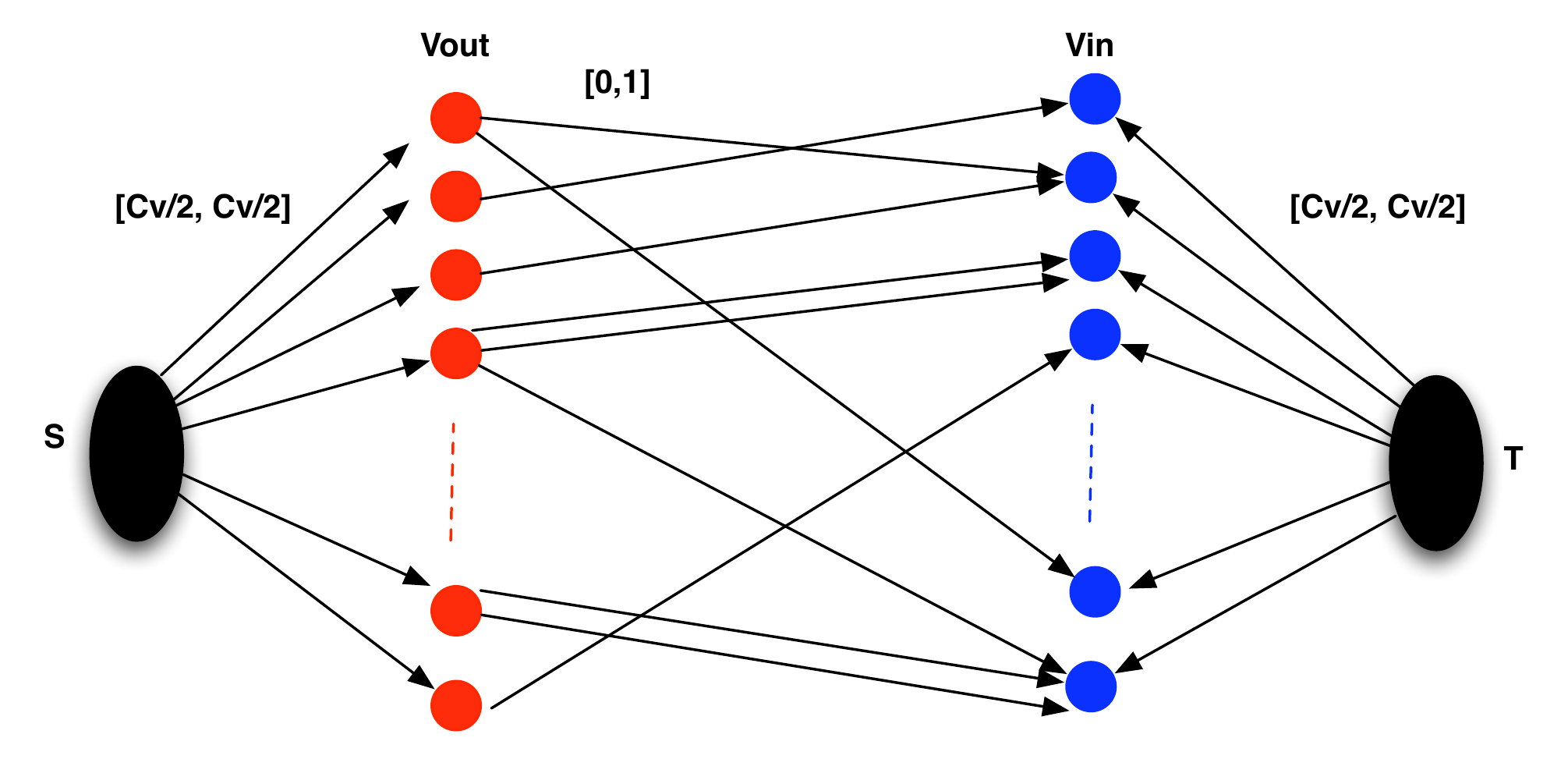}
\caption{Flow Network for Step (4)}
\label{fig:flow}
\end{figure}
\vspace{0.1in}

\vspace{0.1in}
\noindent{Step (1)-(3):} 
The first three steps  are a generalization of
Peterson's theorem. 
$G'$ can be constructed as follows.
For any node $v$ with degree less than $c_v \Delta'$, we add
loops until degree of the node becomes at least $c_v \Delta' - 1$.
Note that after the modification, the number of node with degree 
$c_v \Delta' - 1$ is even as $c_v$'s are even. Pair the nodes and add edges 
so that every node has degree $c_v \Delta'$.

Note that each node in $G'$ has even degree  as all $c_v$'s are even. 
Therefore, we can find a Euler cycle $EC$ on $G'$.
Note that  for each node $v$, there are $c_v \Delta'/2$ incoming edges
and $c_v \Delta'/2$ outgoing edges in $EC$ .

We construct a bipartite graph $H$ by considering the directions of edges
obtained in $EC$. For each node $v$ in $G'$, create two copies $v_{in}$
and $v_{out}$. For an edge $e = (u, v)$ in $G'$, if the edge is visited
from $u$ to $v$ in $EC$, then create an edge from $u_{out}$ to $v_{in}$ in $H$.
As each node $v$ in $G'$ has $c_v \Delta'/2$ incoming edges and
$c_v \Delta'/2$ outgoing edges in $EC$, the degrees of $v_{in}$ and
$v_{out}$  in $H$ is also $c_v \Delta'/2$.

We construct a bipartite graph $H$ by considering the directions of edges
obtained in $EC$. For each node $v$ in $G'$, create two copies $v_{in}$
and $v_{out}$. For an edge $e = (u, v)$ in $G'$, if the edge is visited
from $u$ to $v$ in $EC$, then create an edge from $u_{out}$ to $v_{in}$ in $H$.
As each node $v$ in $G'$ has $c_v \Delta'/2$ incoming edges and
$c_v \Delta'/2$ outgoing edges in $EC$, the degrees of $v_{in}$ and
$v_{out}$  in $H$ is also $c_v \Delta'/2$.

\vspace{0.1in}
\noindent{Step (4):}
We now find a $c_v/2$-matching in $H$  where exactly $c_v/2$ edges are matched for each $v_{in}$
and $v_{out}$. 

\begin{theorem}
We can find an optimal migration schedule when each node has even $c_v$.
\end{theorem}

We can show the theorem by showing the following lemmas 

\begin{lemma}
There exists  a $c_v/2$-matching in $H$ and it can be found in polynomial time.
\end{lemma}

\begin{lemma}
We can decompose $H$ into $M_1, M_2, \dots, M_{\Delta'}$
so that each $M_i$ is a $c_v/2$-matching in $H$.
\end{lemma}

\begin{lemma}
Each component $M_i$ can be scheduled in one round.
\end{lemma}

\subsection{Soft Edge Coloring - General Case}
In this section, we consider the case that each node $v$ has an arbitrary $c_v$.
The problem is \textsf{NP}-hard even when $c_v = 1$ for all nodes.
We develop a soft edge coloring algorithm that colors edges of the given graph so that  
the transfer constraints $c_v$ of the nodes are satisfied. 
The coloring defines a data migration schedule and, 
as the number of colors used determines the number of rounds in our schedule, 
we would like our coloring algorithm to minimize the number of colors needed.
We obtain an algorithm that uses at most $OPT + \sqrt{OPT}$ colors

\subsubsection{Outline of the Algorithm}
We first  give an overview of the coloring algorithm. 
Our algorithm is inspired by the recent work for multi-graph edge coloring algorithm by
Sanders and Steurer~\cite{sanders05} and generalized their algorithm.
Our algorithm uses three particular subgraph structures, balancing orbits, color orbits and edge orbits,
which is defined Section~\ref{sec:struct}. 
The latter two structures --- color orbits and edge orbits --- are generalizations of 
the structures used by Sanders and Steurer~\cite{sanders05}.

The algorithm starts with a naive partial coloring of $G = (V, E)$ and proceeds in two phases. 
In the first phase, we use three structures and color edges until we produce 
a simple uncolored subgraph $G_0$  (Section \ref{sec:bad})
consisting of small connected components (Section \ref{sec:size}); in the second phase 
we color $G_0$ and show that $O(\sqrt{{d_v(G_0)}/{\min c_v}})$ new colors
are enough to obtain a proper coloring in $G_0$ (Section \ref{sec:simple}).

\subsubsection{Preliminaries} \label{sec:struct}

We first introduce some definitions.
Let $|E_i(v)|$ be the number of edges of color $i$ adjacent to a vertex $v$.
\begin{definition}[Strongly/lightly missing color]
Color $c$ is \emph{saturated} at vertex $c$ if $|E_c(v)| = c_v$. The color $c$ is \emph{missing} at vertex $v$ if $|E_c(v)|$ is less than $c_v$; in this case we distinguish two possibilities: 
\begin{itemize}
\item $c$ is \emph{strongly missing} if $|E_c(v)| < c_v -1$. 
\item $c$ is \emph{lightly missing} if $|E_c(v)| = c_v -1$.
\end{itemize}
\end{definition} 


We will reuse definition \ref{def:path} for altenating paths but  unlike the case when $c_v = 1$, an alternating path may not be a simple path  as there can be multiple edges with the same color incident to a node.

\paragraph{Balancing Orbits}
\vspace{5pt}
We first define balancing orbits as follows.
\begin{definition}[balancing orbit] 
\label{def:balancing}
A \emph{balancing orbit} $O$ is a node induced subgraph such that 
all nodes $V(O)$ are connected by uncolored edges and the following property holds
\begin{itemize}
\item A vertex $v \in V(O)$ is strongly missing a color.
\item There are at least two nodes $u,v \in V(O)$ lightly missing the same color.
\end{itemize}
\end{definition}
\begin{figure}
\centering
\includegraphics[scale = 0.45]{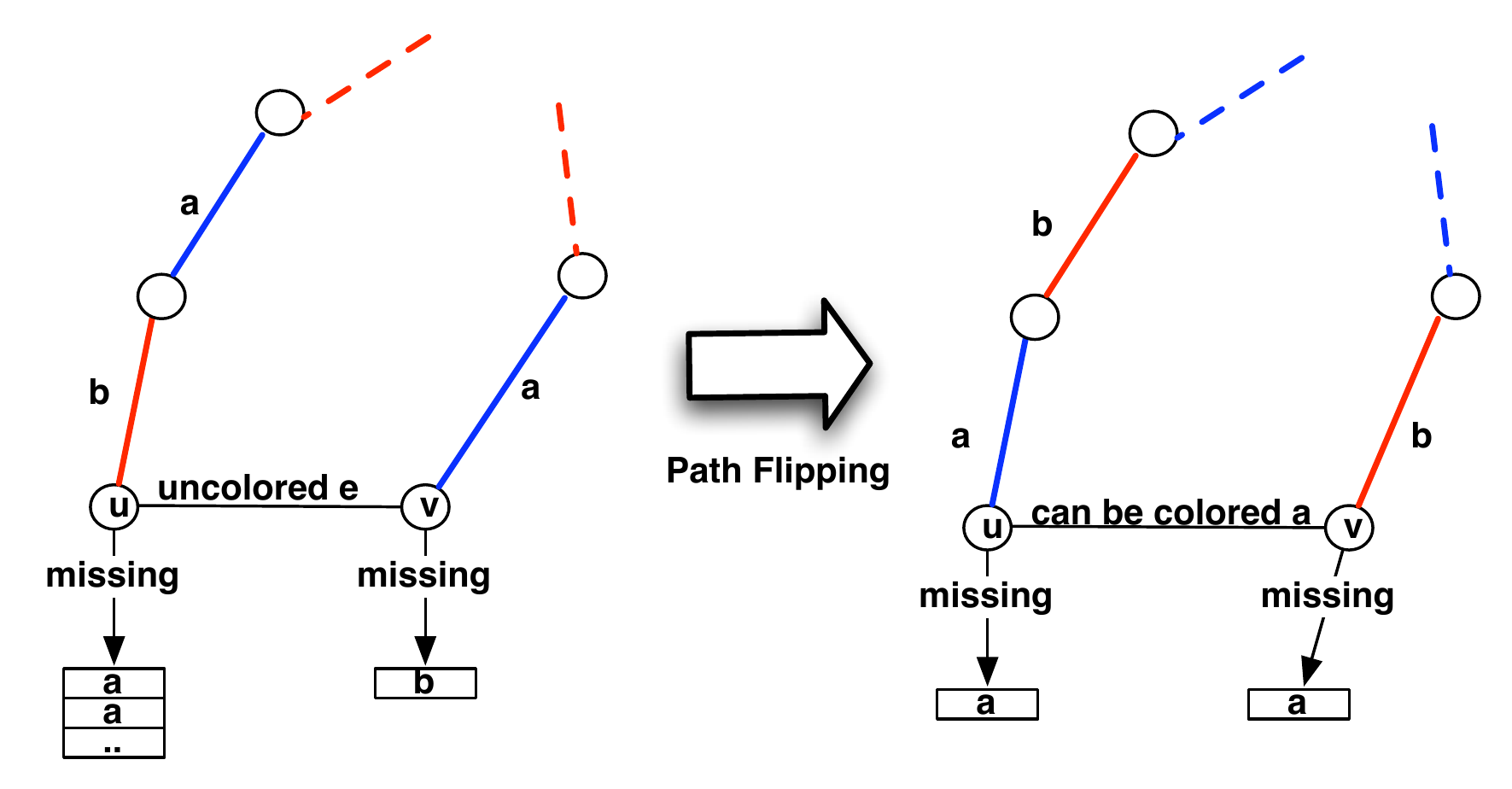}
\caption{$u$ strongly missing $a$ and path $P$ ends at $v$, we can color $e$ with $a$}
\label{fig:bal}
\end{figure}
\vspace{0.2in}
The following lemma shows that if we have a balancing orbit, we can color
an uncolored edge and eventually remove any balancing orbits.
\begin{lemma}\label{lemma:balancing}
If there is a balancing orbit in $G$, 
then we can color a previously uncolored edge.
\end{lemma}

\paragraph{Color Orbits and Edge Orbits}
In this section, we define two subgraph structures: a \emph{color orbit} and an \emph{edge orbit}, 
which are basically generalizations of the structures defined in \cite{sanders05}.

\begin{definition}[Color orbit]
\label{def:color}
A \emph{color orbit} $O$ is a node induced subgraph such that
all nodes $V(O)$ are connected by uncolored edges and the following property holds
\begin{itemize}
\item There are at least two nodes $u,v \in V(O)$ lightly missing the same color.
\end{itemize}

\end{definition}

\begin{lemma}\label{lemma:balancing2}
\cite{sanders05} If there exists a color orbit in $G$ then we can color a previously uncolored edge.
\end{lemma}

By Lemma \ref{lemma:balancing} and \ref{lemma:balancing2}, 
whenever we find a balancing orbit or color orbit, we can color a previously 
uncolored edge and make progress.
If neither of properties in Definition \ref{def:balancing} and \ref{def:color}
hold, we call $O$ a \emph{tight color orbit}.

\vspace{5pt}
Our goal at the end of Phase 1 of the algorithm is to get a \emph{simple} uncolored graph $G_0$ consisting of small connected components.
That is, in $G_0$ there cannot be more than one uncolored edges between two nodes. 
In order to eliminate parallel uncolored edges the following subgraph structure is used.

\begin{definition}[Lean and bad edges]
If an edge $e$ is colored and all its parallel edges are colored then $e$ is a \emph{lean} edge. 
If $e$ is uncolored and has a parallel uncolored edge then $e$ is a \emph{bad} edge.
\end{definition}

\begin{definition}
An \emph{edge orbit} is a subgraph consisting of two uncolored parallel edges (called the \emph{seed} of the edge orbit) and then is inductively defined as follows:
Let $e = (x,y)$ be an edge in the edge orbit $O$, let $a$ and $b$ be missing colors at $x$ and $y$ respectively and let $P$ be the alternating path starting at $x$ then $O \cup P$ is an edge orbit if
\begin{itemize}
\item no edge of color a or b is contained in $O$.
\item $\exists v \in P$ that was not in the vertex set of $O$.\\
\end{itemize}
\end{definition}
If edge orbit $O$ has a lean edge then $O$ is called a weak edge orbit otherwise $O$ is a tight edge orbit. A color $c$ is free for an edge orbit $O$ if $O$ does not contain an edge with color $c$. 

The following lemma from \cite{sanders05} states that if in some coloring of $G$, there exists a weak edge orbit then we make progress toward our goal of obtaining $G_0$ by either coloring a previously uncolored edge or by uncoloring a lean edge and coloring a bad edge.
\begin{lemma}\label{lemma:weakedge}
\cite{sanders05} If a coloring of $G$ contains a weak edge orbit then we can either color a previously uncolored edge or we can uncolor a lean edge and color a bad edge.\\
\end{lemma}

A tight edge orbit does not have lean edges so its vertex set is connected 
by uncolored edges and thus a tight edge orbit is one of the following 
--- a balancing orbit, color orbit or a tight color orbit. 
When it is a tight color orbit, as we cannot make progress toward $G_0$, 
which we call a \emph{hard orbit}. 
Note that no vertex in a hard orbit is strongly missing a color, no two
nodes are lightly messing the same color, and no edge in a hard orbit is lean.

\paragraph{Growing Orbits}
A color $c$ is \emph{full} in a hard orbit $O$ if $c$ is saturated on all vertices of $V(O)$ but at most one vertex in $V(O)$ is lightly missing $c$ or equivalently if $|E_c \cap E(V(O))| \geq  {\lfloor \frac{\sum_{v \in V(O)} c_v}{2} \rfloor}. $
So if color $c$ is full in a hard orbit $O$ it cannot be used to color uncolored edges whose endpoints are in $O$.
\begin{definition}[Lower bound witnesses]
A hard orbit is a $\Delta'$-witness if all missing colors at some node are non-free. It is a $\Gamma'$-witness if all free colors of the orbit are full.
\end{definition}
The intuition behind the witnesses is the following. Suppose very few colors are used in hard orbit $O$, in the case of $\Gamma'$-witness almost all color classes are full in $O$ and in the case of a $\Delta'$-witness almost all available colors are strong on some node $v \in V(O)$. So a witness in some coloring using $q$ colors indicates that it is \textit{almost} impossible to color an additional uncolored edge using the available $q$ colors and thus the number of available colors needs to be increased.  
\begin{lemma}\label{lemma:groworbit}
\cite{sanders05} Given a hard orbit in some coloring we can either find a witness or compute a larger edge orbit.
\end{lemma} 

\subsubsection{Algorithm}\label{sec:alg}
The algorithm proceeds in two phases. The outcome of the first phase would be $G_0$, a simple uncolored graph with no large components. 
The following procedure for the first phase 
eliminates all the bad edges in $G$ (Section \ref{sec:bad}) and reduce the size of connected components 
(Section \ref{sec:size}), which gives $G_0$ with the desired properties.
In the second phase (Section \ref{sec:simple}), we color the remaining subgraph $G_0$.

\paragraph{Eliminating bad edges}
\label{sec:bad}
Given a partial coloring using $q$ colors, we iterate over a list of bad edges and we execute the following steps (Note a bad edge is a trivial edge orbit).
Given an edge orbit $O$ 
\begin{enumerate}
\item[(1)] If nodes of $O$ form a balancing or color orbit, 
apply Lemma \ref{lemma:balancing} or \ref{lemma:balancing2}.
\item[(2)] If $O$ is weak, apply Lemma \ref{lemma:weakedge}.
\item[(3)] If $O$ is a hard orbit, apply Lemma \ref{lemma:groworbit}.
\begin{enumerate}
\item If Lemma \ref{lemma:groworbit} gives a larger edge orbit $O \cup P$, repeat with $O = O \cup P$.
\item If Lemma \ref{lemma:groworbit} gives a witness then increase $q$ by one color and color the bad edges in the seed with the additional color.
\end{enumerate}
\end{enumerate}
\vspace{11 pt}
The output of this procedure is a simple subgraph $G'$ of $G$ induced by uncolored edges. 
In Lemma \ref{lemma:witness1} and Lemma \ref{lemma:witness2},
we show an upper bound on the number of used colors if there is a $\Delta'$ or $\Gamma'$-witness. The next procedure reduces the size of the connected components of $G'$ whenever $G'$ has balancing or color orbits.

\paragraph{Reducing size of connected components}
\label{sec:size}
For every connected component $U$ of $G'$, 
\begin{enumerate}
\item If $U$ contains a vertex that is strongly missing a color then use Lemma \ref{lemma:balancing} to color an uncolored edge.
\item If $U$ contains two or more vertices that are lightly missing the same color use Lemma \ref{lemma:balancing2} to color an uncolored edge.
\end{enumerate}
So at the end of the first phase we have the simple subgraph $G_0$ where for every connected component $U$ of $G_0$,
 no vertex is strongly missing a color and no two vertices of $U$ miss the same color. 
 In Lemma \ref{lemma:size}, we show that the size of $G_0$ is no more than $\frac{q + 2}{q  - \Delta' + 2}$.
\paragraph{Coloring $G_0$} \label{sec:simple}
Phase 2 colors $G_0$. We use  only $ \max_v \lceil \frac{d_v(G_0)}{c_v}  \rceil + 1 $ colors. The procedure goes as follows: 
\begin{enumerate}
\item Create $c_v$ copies of each vertex $v$ and distribute the edges over the copies so that each vertex is adjacent to at most 
$\lceil \frac{d_v(G_0)}{c_v} \rceil $ edges where $d_v(G_0)$ represents the degree of $v$ in $G_0$.
\item Use Vizing's algorithm to properly color each component. We need at most $\max_v \lceil \frac{d_v(G_0)}{c_v} \rceil $ + 1 colors.
\item Contract the copies back to $v$ getting a coloring where for any node $v$ there is no more than
$c_v$ edges of the same color. 
\end{enumerate}

\subsubsection{Analysis}\label{sec:analysis}
In the following $q$ denotes the total number of colors available for the algorithm. We show that the algorithm colors all the edges of $G$ using at most $q = OPT + \Theta(\sqrt{OPT})$ colors.
We first bound the number of used colors when there is a $\Delta'$ or $\Gamma'$-witness.
In particular, we show that when there exists a $\Gamma'$-witness, 
the total number of colors is a constant more than $OPT$ and does not depend on the size of $|V(O)|$,
which is tighter than the analysis given in \cite{sanders05}.

\begin{lemma}\label{lemma:witness1}
Let $O$ be a hard orbit. If $O$ is a $\Delta'$-witness then
$ q \leq \Delta' + \frac{2|V(O)| -4}{c^-}$
where  $c^- = \min_{v \in V(O)} c_v$. 
\end{lemma}

\begin{lemma}\label{lemma:witness2}
Let $O$ be a hard orbit. If $O$ is a $\Gamma'$-witness then $q \leq \Gamma' + 2|V(O)| - 4 - \frac{2}{c^+}\,.$
\end{lemma}

We now bound the size of $G_0$.

\begin{lemma} \label{lemma:size}
Let $O$ be a tight color orbit. Then $|V(O)| \leq \frac{q + 2}{q  - \Delta' + 2}\,.$ 
\end{lemma} 

The following corollary follows from Lemma \ref{lemma:witness1}, \ref{lemma:witness2}, 
and Corollary \ref{cor:size}, 

\begin{corollary}\label{cor:cwitness}
If $q = \lfloor (1 + \epsilon) \Delta' \rfloor - 1$ and 
there is a witness then $q \leq OPT + \frac{2}{\epsilon } - 2$
\end{corollary}

The following lemma provides a bound on the number of required colors for $G_0$.
\begin{lemma} \label{lemma:G0}
Suppose that the size of the largest component of $G_0$ is bounded by $C$. 
Then coloring $G_0$ requires at most $\lceil \frac{C-1}{c^-} \rceil  + 1$ colors.
\end{lemma}

\begin{theorem}
Given a transfer graph $G$, we can compute a coloring of the edges using at most $OPT +  O\left(\sqrt{OPT}\right)$
colors.
\end{theorem}

\begin{corollary}
The coloring algorithm uses at most $OPT +  O(\sqrt{OPT})$ colors, which implies an approximation factor of $1 + o(1)$ as  $OPT$ increases.
\end{corollary}

%
\bibliographystyle{plain}
\bibliography{TR_SoftEdge}
 
\end{document}